\documentclass[conference]{IEEEtran}
\IEEEoverridecommandlockouts
\usepackage{cite}
\usepackage{amsmath,amssymb,amsfonts}

\usepackage{graphicx}
\usepackage{array}

\usepackage{pgfplots}
\usepackage{tikz}

\usepackage{textcomp}
\usepackage{xcolor}
\usepackage{textcomp}
\usepackage{subcaption}
\usepackage{booktabs}

\usepackage{url}

\usepackage{xcolor}
\usepackage[normalem]{ulem}

\usepackage{algorithm}
\usepackage{algpseudocode}

\usepackage{dsfont}

\usepackage{amsthm}
\newtheorem{theorem}{Theorem}
\newtheorem{lemma}{Lemma}

\newcommand{\prot}{{\sc D2M}}
\newcommand{\cone}{{\sc CONE}}

\def\BibTeX{{\rm B\kern-.05em{\sc i\kern-.025em b}\kern-.08em
    T\kern-.1667em\lower.7ex\hbox{E}\kern-.125emX}}

\newenvironment{customlegend}[1][]{%
    \begingroup
    \csname pgfplots@init@cleared@structures\endcsname
    \pgfplotsset{#1}%
}{%
    \csname pgfplots@createlegend\endcsname
    \endgroup
}%
\def\addlegendimage{\csname pgfplots@addlegendimage\endcsname}

\begin{document}

\title{D2M: A Decentralized, Privacy-Preserving, Incentive-Compatible Data Marketplace for Collaborative Learning\\
}

\author{\IEEEauthorblockN{
Yash Srivastava\IEEEauthorrefmark{2},
Shalin Jain\IEEEauthorrefmark{2},
Sneha Awathare\IEEEauthorrefmark{1} and 
Nitin Awathare\IEEEauthorrefmark{2} }

\IEEEauthorblockA{\IEEEauthorrefmark{2} Department of Computer Science and Engineering, Indian Institute of Technology, Jodhpur.}
\IEEEauthorblockA{\IEEEauthorrefmark{1} School of Mathematics and Computational Sciences, Eastern University Pennsylvania.}
\IEEEauthorblockA{\{shrivastava.12, jain.75, nitina\}@iitj.ac.in, nitina@cse.iitb.ac.in, sneha.awathare@eastern.edu}}



\maketitle

\begin{abstract}
The rising demand for collaborative machine learning and data analytics calls for secure and decentralized data sharing frameworks that balance privacy, trust, and incentives. Existing approaches, including federated learning (FL) and blockchain-based data markets, fall short: FL often depends on trusted aggregators and lacks Byzantine robustness, while blockchain frameworks struggle with computation-intensive training and incentive integration.  

We present \prot, a decentralized data marketplace that unifies federated learning, blockchain arbitration, and economic incentives into a single framework for privacy-preserving data sharing. \prot\ enables data buyers to submit bid-based requests via blockchain smart contracts, which manage auctions, escrow, and dispute resolution. Computationally intensive training is delegated to \cone\ (\uline{Co}mpute \uline{N}etwork for \uline{E}xecution), an off-chain distributed execution layer. To safeguard against adversarial behavior, \prot\ integrates a modified YODA protocol with exponentially growing execution sets for resilient consensus, and introduces Corrected OSMD to mitigate malicious or low-quality contributions from sellers. All protocols are incentive-compatible, and our game-theoretic analysis establishes honesty as the dominant strategy.  

We implement \prot\ on Ethereum and evaluate it over benchmark datasets---MNIST, Fashion-MNIST, and CIFAR-10---under varying adversarial settings. \prot\ achieves up to 99\% accuracy on MNIST and 90\% on Fashion-MNIST, with less than 3\% degradation up to 30\% Byzantine nodes, and 56\% accuracy on CIFAR-10 despite its complexity. Our results show that \prot\ ensures privacy, maintains robustness under adversarial conditions, and scales efficiently with the number of participants, making it a practical foundation for real-world decentralized data sharing.  
\end{abstract}



\section{Introduction}
\label{sec:intro}
The growing demand for collaborative machine learning and data analytics has highlighted the need for \textit{secure, decentralized data sharing}~\cite{kairouz2021advancesopenproblemsfederated, ChiangStrategy}. The applications that need data sharing range from cross-silo healthcare~\cite{Rieke2020} and financial modeling~\cite{abbas} to IoT and smart city analytics~\cite{edge19}, where raw data are highly sensitive. To facilitate this data market platforms are proposed where data owners share and monetize assets via token incentives~\cite{CastroDMP, Agarwal2019, DRIESSEN}.  However, they centralize data in the cloud, raise privacy and trust concerns~\cite{challengesclinicaldata}.

Federated learning (FL) has emerged to mitigate this by keeping data local and sharing only model updates~\cite{yang2019federatedmachinelearningconcept}. In this regard, FL research has yielded many optimizations (e.g., adaptive client sampling) and privacy mechanisms.  Recent work such as UnifyFL~\cite{s2025unifyflenablingdecentralizedcrosssilo}, martFL~\cite{li2023martfl} and Dealer~\cite{liu2021dealer} introduce economic models and privacy guarantees in FL settings. However, majorities of FL systems 
still rely on a \textit{trusted aggregator} and are vulnerable to poisoning or privacy attacks~\cite{banerjee2020}. Furthermore, they lacks \textit{Byzantine robustness}, or rely on central trust, or treat incentives separately from computation, i.e., getting highest economic gains has been neglected which should have incentivizes everyone to participate in such a data market/sharing platform. 
Furthermore, majority of these works suffer other inefficiencies such as uniform sampling and expensive Shapley-value payment calculations~\cite{CastroDSC}.

In contrast, blockchain-based frameworks eliminate central intermediaries for lightweight computations but face challenges in supporting computation-intensive training and sophisticated incentive schemes~\cite{Privado}. For example, \cite{Zyskind2015} proposed \textit{Enigma}, a peer-to-peer multi-party computation(MPC) platform with a blockchain controller to enable privacy-guaranteed computation. \cite{Dong} designed a blockchain-backed IoT data trading platform using Ethereum smart contracts for escrow and periodic checkpoints. \cite{banerjee2020} build a decentralized content marketplace using Tahoe-LAFS and Hyperledger Fabric.

Given this, we can conclude that a fully end-to-end system enabling secure, privacy-preserving data sharing-while supporting the execution of computationally intensive tasks and ensuring robustness against Byzantine users-does not yet exist.

In this paper, we propose \textbf{\prot}, a decentralized data marketplace that unifies federated learning, economic incentives, and blockchain\footnote{Note that \prot\ is compatible with any distributed system; however, we choose blockchain to provide more transparent and reliable accounting of monetary transactions among users—in our case buyers, sellers, and \cone\ nodes (responsible for executing computationally intensive tasks such as model training in a federated setup).} arbitration into one umbrella achieving the goal of secure, privacy-preserving data sharing. The core idea of \textbf{\prot} is to enable distributed model training using data from multiple owners, in a way that meets the requirements of the user or data buyer—without exposing the owners' raw data or the buyer's complete model.

In this regards, the data buyer submits a bid-based request on the blockchain, which is processed through an auction. A blockchain arbitrator escrows payments and mediates disputes for the data request. Since current blockchains lack support for computationally intensive tasks~\cite{tuxedo, renoir} such as training machine learning models, execution is carried out off-chain on a distributed system called \cone\ (\uline{Co}mpute \uline{N}etwork for \uline{E}xecution). To protect \prot\ from Byzantine behavior in \cone, a randomized set of compute nodes is selected using \textit{YODA}~\cite{DasYODA}, but unlike naive YODA, we employ an exponentially growing execution set across rounds. YODA, through MiRACLE’s multi-round hypothesis testing, ensures consensus even in the presence of adversaries. The use of an exponentially increasing execution set yields faster convergence than a fixed-size set.

To ensures the safeguard against the model discloser, \prot\ performs the first level of the model training at the seller side while the global training takes place at the \cone\ that aggregates the weights of the local models using federated learning setup. 
These sellers are selected by the blockchain contract based on the data requests placed by the data buyer. In other words, only those sellers which holds the data relevant to the buyers requests are involved in the process. In our setting, sellers may act maliciously by performing the initial (local) training using corrupted data, or they may simply possess low-quality data that does not align with the buyer’s requirements. To address this, we propose Corrected OSMD, which combines the strengths of OSMD~\cite{ZhaoOSMD} and Corrected KRUM~\cite{DuBlockchainDML}. Corrected OSMD effectively mitigates the influence of malicious sellers.

Once the training accomplish each seller earns a rewards to provide the data. To make it fair \prot\ employs fair distribution of payments to sellers which has most relevant data, which is achieved with the adaption of the recent online mirror descent client-sampling method \cite{ZhaoOSMD}. 
It is important to note that, all protocols in \prot\ are incentive-compatible. Using which we prove that \textbf{\prot} ensures a game theoretic security, meaning any deviating party would incurs a loss in term of token earned. 

Lastly, we have implemented \prot\ in a fully distributed setup with Ethereum~\cite{wood2014ethereum} 
as the underlying blockchain network consisting of 28 nodes. We extensively evaluate \prot\ on benchmark datasets--MNIST~\cite{mnist}, CIFAR-10~\cite{cifar}, and Fashion-MNIST~\cite{fashionmnist} with different number of sellers (ranging from 50 to 200) and varying proportions of malicious \cone\ nodes.

\prot\ achieves up to 98.75\% accuracy on MNIST and 90.13\% on Fashion-MNIST with minor drops ($<$3\%) till 30\% byzantine nodes. On CIFAR-10, accuracy reaches 56.5\% with similar degradation under adversarial settings. \prot\ also exhibits linear growth of overhead with increasing number of participants and also does not lead to excessive increase in rounds hence time with increasing number of malicious actors. These results confirm \prot's ability to preserve privacy, ensure robustness, and scale to real-world decentralized data sharing scenarios.

\subsection{Key Contributions}
We summarize the novel aspects of \prot\ as follows:
\begin{itemize}

\item We propose \prot, enabling on-chain auctions for privacy-preserving data sharing., which consists of Buyer, Seller, and \cone\ nodes.  

\item To counter malicious sellers providing low-quality data, we introduce Corrected OSMD, and to address adversarial \cone\ nodes, we employ a modified version of YODA.

\item A game-theoretic analysis formally demonstrates that honesty prevails as the dominant strategy.  

\item Finally, we implement \prot\ in a fully distributed setup using Ethereum~\cite{wood2014ethereum} as the underlying blockchain, and evaluate it on three benchmark datasets, demonstrating the efficiency of \prot.

\end{itemize}

The rest of the paper is organized as follows: Section \ref{sec:SystemDesign} presents the system design of \prot, including its components and interactions. Section \ref{sec:offchain} details off-chain computation using \cone\ nodes and Corrected OSMD. Section \ref{sec:game} provides the game-theoretic proof of honest dominance. Section \ref{sec:evaluation} reports our evaluation and observations. Section \ref{sec:related} reviews related work. Finally, Section \ref{sec:futureWork} concludes with limitations and future directions.

\section{System Design, Notations, and Workflow}
\label{sec:SystemDesign}

\prot\ is a data marketplace that facilitate the distributed training of a machine learning model without revealing data as well as a complete model to the non-owners. \prot\ is primarily 
designed with four entities--data recipients/buyers who need data for their machine learning tasks, data providers/sellers who own and provide data, arbiters (blockchain) to handle payments and disputes, and, \cone\ to handle computationally-intensive tasks. Each of these entities interact with each other through the blockchain transactions. The overall workflow that shows interaction between these entities proceeds in two phases:
\begin{enumerate}
  \item \emph{On‐chain Auction Phase}:  
    Buyers post bids on the blockchain; matching bids within a \(T\)-block window enter a transparent auction; the highest bid wins; and eligible sellers are identified via the data registry.
  \item \emph{Off‐chain Computation Phase}:  
    A YODA‐selected execution set orchestrates iterative, privacy‐preserving model training.  Sellers compute local updates, compute nodes aggregate local updates via OSMD + Corrected Krum, and MIRACLE ensures Byzantine‐robust consensus.  Iterations continue until the buyer’s target metric is met; then the final model is delivered on‐chain and funds are automatically disbursed.
\end{enumerate}

The successful execution of these two phases involves interactions among multiple entities, which we discuss in detail next.

\subsection{System Actors and their notations} 
As mentioned earlier in  \prot's consists of Buyers, Sellers, Arbiters, and \cone. A brief detail of each is given below.

\begin{itemize}
    \item \textbf{Buyers}: Buyers, represented by $\mathcal{B}=\{\mathcal{B}_1, \mathcal{B}_2, \cdots, \mathcal{B}_N\}$, are end users of \prot\ who require datasets to meet specific needs, such
    as data attributes, size, quality, or geographical relevance. They initiate the process by submitting a bid $\mathcal{R}$, defined by four parameters: the bid amount $Amt$; metadata (tags) $\xi$ for retrieving relevant datasets; the model $\mathcal{K}$, which is iteratively trained; and a metric $M$ with threshold $\tau$, used to evaluate the model’s quality and determine the extent of training.

    \item \textbf{Sellers}: Sellers, denoted by $\mathcal{S} = \{\mathcal{S}_1, \mathcal{S}_2, \cdots, \mathcal{S}_M\}$, are another end user of the \prot\ that owns datasets relevant to the buyers’ requirements. These datasets can range from structured and unstructured data to specialized domain-specific collections. Sellers participate by offering their data in response to buyer bids. They are incentivized through an auction-based payment mechanism (facilitated by Arbitrage-explained next), where the revenue is directly tied to the quality and relevance of their data.

    \item \textbf{Arbitrage}: Arbitrage, represented by $\mathcal{A}$, is fully distributed blockchain system that guarantee accuracy, fairness, and openness in all data exchanges. Furthermore, Arbitrage is responsible for confirming the integrity of transactions between buyers and sellers, validating bids, and settling disagreements. Moreover, Arbitrage ensure fair decisions under Byzantine fault-tolerant consensus and enforce buyer-defined evaluation metrics. Note that \prot\ would work with any other distributed system. However, we choose blockchain for better accounting of the monitory transactions between buyers, sellers, and \cone\ node.

    \item \textbf{Compute Network (\cone)}: \cone, denoted by $\mathcal{C}$, is a decentralized system for executing computationally intensive tasks, specifically multiple rounds of the Corrected OSMD algorithm, which preserves sellers’ data privacy. Details of this algorithm and its role in the system are provided in Section~\ref{sec:workflow}.

\end{itemize}

 \begin{figure}
    \centering
    \includegraphics[width=0.68\linewidth]{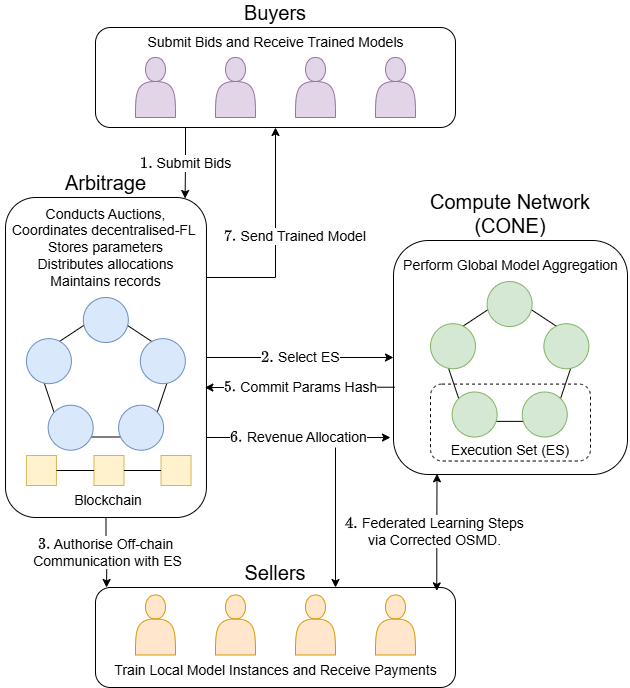}
    \caption{System actors and interactions in \prot. Buyers and Sellers interact through Arbitrage and Compute Nodes to facilitate secure, privacy-preserving computation on decentralized data.}
    \label{fig:dist-design}
    \end{figure}

\prot\ facilitate interaction between end users, Buyers and Sellers, through Arbitrage and Compute nodes, to establish a secure, transparent, and decentralized marketplace for data sharing. Next we will discuss this interaction in details.

\subsection{System Workflow} 
\label{sec:workflow}
To commence the workflow any Buyer, say $\mathcal{B}_i$ place a data request (bid) \(\mathcal{R} = \{\xi, Amt, \mathcal{K}, M, \tau\}\) on to the blockchain. The $\mathcal{R}$ is placed onto the blockchain using a blockchain transaction, say $tx_{\mathcal{B}_i}$, which will trigger a method ($startAuction$) in \prot's smart contract (Algorithm~\ref{alg:auction}) that we assumed to be already deployed before the workflow is commenced. 

If an auction for specific tags has already been initiated by another buyer and is still active, a new buyer can place a bid using $placeBid$. When the auction closes (i.e., after a timeout), the highest bidder wins. The contract then invokes the $IdentifyMatchingDatasets$ method to select sellers whose data matches tag $\xi$. Following this, off-chain computation begins and proceeds in rounds, which is facilitated by Algorithm~\ref{alg:workflow}. In each round, it selects a subset of \cone\ nodes using the YODA sortition algorithm. To mitigate the influence of malicious CONE nodes, YODA progressively increases the subset size until a majority consensus is reached.

The off-chain computation involves aggregating model weights trained locally by individual sellers whose datasets match the tags $\xi$. Existing approaches typically use one of two methods: OSMD or Corrected KRUM. 

The OSMD algorithm enables the integration of multiple sellers by ensuring fair budget allocation and high-quality data aggregation through adaptive sampling based on utility maximization. However, despite its utility-driven sampling, OSMD averages the updates from all sampled sellers. As a result, if any Byzantine (malicious) sellers are included in the sample, their faulty updates also influence the final model.


Corrected KRUM, on the other hand, selects a single update that is closest to the mean effect of all sellers' updates in terms of krum scores. However, relying solely on Corrected KRUM poses a drawback: it does not account for data quality. If most sellers have low-quality data—for instance, 10 with poor data and only 3 with high-quality data—the mean skews toward the poor updates, and Corrected KRUM may end up selecting a suboptimal update.

To address these limitations, we combine OSMD and Corrected KRUM, called corrected OSMD. Original OSMD first adaptively samples the most promising sellers based on utility, ensuring a focus on high-quality data. Then, instead of averaging updates from all sampled sellers (as in standard OSMD), Corrected KRUM selects the single update closest to the mean among this refined subset, mitigating the influence of Byzantine sellers. Moreover, OSMD supports fair revenue distribution by evaluating each seller's utility contribution—something Corrected KRUM alone does not offer. This hybrid approach ensures both robustness against Byzantine behavior and equitable compensation based on data quality.


\begin{algorithm}[H]
\caption{\prot \ AuctionContract ($C_{\mathrm{Auction}}$)}
\label{alg:auction}
\begin{algorithmic}[1]

\State {\bf Variables: }
    \State $buyers \gets \{\}$, $sellers \gets \{\}$, 
    \State \(activeAuctions<AuctionEnd, HighestBid, highe-$\\$stBidder> \gets \{\}$

\State {\bf Functions: }
\Procedure{userRegistration}{$userId, isBuyer$}
\If{$isBuyer$}
    $buyers[caller]=userId$
\Else  
    \text{ }$sellers[caller]=userId$
\EndIf
\EndProcedure\\

{\color {gray}//bid.$\xi$ and $B_{cur}$ denote metadata—tags and current block height, respectively. The procedure initiates an auction.}
\Procedure{startAuction}{bid}
  \If{$caller$ is buyer \textbf{and} no auctions with tag $bid.\xi$}
      \State $activeAuctions[bid.\xi] \gets \text{ } <B_{cur}+T,0,"">$
  \EndIf
  \State \Call{placeBid}{bid}
\EndProcedure\\
{\color {gray}//Place the bid in a new or existing auction; if the active auction tagged $bid.\xi$ is ongoing and gets a higher bid.}
\Procedure{placeBid}{bid}
  \If{Auction tagged $bid.\xi$ is active and receives a bid higher than the current maximum}

      \State $activeAuctions[bid.\xi].highestBid \gets bid.Amt$ 
      \State $activeAuctions[bid.\xi].highestBidder \gets caller$

  \EndIf
\EndProcedure
\State {\color {gray}// This method  will terminate the auction}
\Procedure{closeAuction}{tags}
  \If{Auction tagged $bid.\xi$ ends}
    \State $\mathcal{S}$ $\leftarrow$ \Call{IdentifyMatchingDatasets}{tags}
    \State \textbf{/*Invoke off-chain compute via \prot\ Core contract in Algorithm~\ref{alg:workflow}, returning seller's contribution ($\widehat{N}$) and CONE node's participation count ($\widehat{NC}$).*/}
    \State $amt=activeAuctions[bid.\xi].highestBid$
    \State
    \Call{onComputeComplete}{$amt$, $\widehat{N}$, $\widehat{NC}$}
  \EndIf
\EndProcedure

\State {\color {gray}//Revenue distribution post bid termination}
\Procedure{onComputeComplete}{bidAmt, sellerContribs, nodeCounts}
  \State $CONEShare$ $\leftarrow$ $bidAmt$ \(\times\) $0.30$
  \State $sellerShare$ $\leftarrow$ $bidAmt$ - $CONEShare$
  \State \(W_n \leftarrow \sum_{\text{(nodeId, c)} \in nodeCounts} c\)  
  \State \(W_s \leftarrow \sum_{\text{(sellerId, w)} \in sellerContribs} w\)
  \ForAll{($nodeId$, $c$) in $nodeCounts$}
    \State Transfer $CONEShare \times c / W_n$ to CONE node with ID $nodeId$.
  \EndFor
  \ForAll{(seller, w) in sellerContribs}
    \State Transfer $sellerShare$ $\times$ w / $W_s$) to seller with ID $sellerID$
  \EndFor
\EndProcedure

\State {\color {gray}//Return sellers set with dataset that matches the tags}
\Procedure{IdentifyMatchingDatasets}{tags}
  \State \Return \(\{\,ds\in\mathrm{datasets}\mid tags\subseteq ds.tags\}\)
\EndProcedure
\end{algorithmic}
\end{algorithm}

In summary, YODA combined with Corrected OSMD facilitates robust off-chain computation. YODA ensures resilience against Byzantine CONE nodes, while the Corrected OSMD hybrid defends against Byzantine sellers. The detailed explanation of each component-YODA and Corrected OSMD-is provided in Section~\ref{sec:offchain}.

Once the off-chain computation is completed, revenue is distributed using the $onComputeComplete$ method (Algorithm~\ref{alg:auction}). \prot\ allocates $30\%$ of the bid amount to \cone\ nodes for their computational support, while the remaining $70\%$ is distributed among the sellers for contributing data (lines 35–36 in Algorithm~\ref{alg:auction}). Seller earnings are further normalized by the quality of their data, as evaluated using the metric $M$ specified by the winning buyer. The workflow concludes with the distribution and disbursement of revenue.
\section{Off-chain Computation-YODA and Corrected OSMD}
\label{sec:offchain}
Having briefly introduced purpose of YODA in Algorithm~\ref{alg:workflow} and Corrected OSMD (Algorithm~\ref{alg:osmd}), as a part of off-chain computation, in the previous section, we now delve into their detailed workings. Additionally, we examine how these two algorithms interact via the blockchain to support \prot's overall objectives.

\subsection{\prot\ Core}
\label{sec:D2Mcore}
In this section, we detail the workings of \prot\ Core, which is activated once the on-chain auction determines the winning bid \(R^* = \{\xi,\; \text{Amt},\; \mathcal{K},\; M,\; \tau\}\) and selects the set of sellers \(\mathcal{S}\) whose datasets match the tags in \(\xi\). The primary objective of \prot\ Core is to produce a final model \(\widehat{\mathcal{K}}\), while ensuring three key goals: tolerating up to \(f_{\max}\) malicious \cone\ nodes, minimizing both communication rounds and computational overhead, and enabling precise reward allocation to each participant in proportion to their verified contribution to the computation.

Before we walk through the algorithm step by step, let us first clarify the notations used throughout. Here, \(t_{\max}\) denotes the maximum number of allowed global rounds, while \(w^t\) represents the global model weights at round \(t\). The sampling probabilities over the seller set \(\mathcal{S}\) in round \(t\) are denoted by \(p^t\). The cumulative seller contribution vector in round \(t\) is given by \(N^t\), and \(NC\) tracks the cumulative participation of compute nodes across rounds. The likelihood threshold required to reach consensus is represented by \(\theta\). Each mini-round \(i\) selects an execution set \(ES_i\), and \(L_{k,i}\) indicates the likelihood score of digest \(k\) after mini-round \(i\). Finally, \(\widehat{\mathcal{K}}\), \(\widehat{N}\), and \(\widehat{NC}\) denote the final trained model, the aggregated seller contributions, and the overall participation record of the \texttt{CONE} nodes, respectively.

We begin by initializing the round counter \(t\), model weights \(w^0\), uniform seller probabilities \(p^0\), and zeroed contribution counters \(N^0,NC\).  This creates a neutral starting state with no bias.

The outer loop (line 4 in Algorithm~\ref{alg:workflow}) continues until the model’s performance \(M(w^t)\) meets the buyer’s target \(\tau\) or we exhaust \(t_{\max}\) rounds-this bounds cost while striving for required quality. At the start of each round (outer loop), we initialize a mini‐round index \(i=0\) and reset all likelihood scores \(L_{k,0}=0\).  These scores will measure how strongly compute nodes agree on each candidate update.

The inner loop (line 6 in Algorithm~\ref{alg:workflow}) forms an Execution Set \(ES_i\) of size \(\mathit{size}_{i-1}+2^{i-1}\). 
By starting small and doubling, we minimize communication when nodes quickly agree, yet guarantee an honest majority will eventually dominate even if some nodes are malicious, which is inferred from YODA's security analysis.

Each node \(j\in ES_i\) logs its participation \(NC[j]\mathrel{{+}{=}}1\) (line 10 in Algorithm~\ref{alg:workflow}) for fair compute‐reward accounting, then runs the Corrected OSMD routine on \((w^t,p^t)\).  Corrected OSMD (detailed in Section~\ref{sec:osmd}) adaptively selects high‐value sellers, combines their updates into candidate weights \(w^t_{i,j}\), revises sampling probabilities \(p^t_{i,j}\), and updates per‐seller contributions \(N^t_{i,j}\).  The node then commits a cryptographic hash \(H(\cdot)\) of its result on‐chain (step 5 in Figure~\ref{fig:dist-design}), preventing any post‐hoc change and enabling verifiable audit during the next iteration.

\begin{algorithm}[H]
\caption{\prot\ Core Algorithm}
\label{alg:workflow}
\begin{algorithmic}[1]
\State \textbf{Input:} \(R^*,\,\mathcal{S},\,t_{\max}\)
\State \textbf{Output:} \(\widehat{\mathcal{K}},\,\widehat{N},\,\widehat{NC}\)
\State Initialize:
\[
  t\leftarrow0,\;
  w^0\leftarrow\mathcal{K},\;
  p^0\leftarrow\tfrac{1}{|\mathcal{S}|}\mathbf1,\;
  N^0\leftarrow\mathbf0,\;
  NC\leftarrow\mathbf0.
\]
\While{\(M(w^t)<\tau\) \textbf{and} \(t<t_{\max}\)}
  \State \(i\leftarrow0\); \(k\leftarrow 0 ; L_{0,0}\leftarrow0\)
  \While{\(\forall k:\;L_{k,i}\le\theta\)}
    \State \(i\leftarrow i+1\)
    \State \(ES_i\leftarrow\text{YODA sortition(size}=\mathit{base}+2^{i-1})\)
    \ForAll{\(j\in ES_i\)}
      \State \(NC[j]\mathrel{+}=1\)
      \State \(\{w^t_{i,j},p^t_{i,j},N^t_{i,j}\}\leftarrow\) Algorithm~\ref{alg:osmd}\((w^t,p^t)\)
      \State Commit \(H(w^t_{i,j},p^t_{i,j},N^t_{i,j})\) on‐chain
    \EndFor
    \ForAll{digest \(k\)}
      \State \(L_{k,i}\leftarrow\sum_{l=1}^i (2\,c_{k,l}-C_l)\,C_l\),\\
      where \(c_{k,l}\) = \# nodes in \(ES_l\) with hash \(k\), \(C_l=|ES_l|\)
    \EndFor
  \EndWhile
  \State Select \(k'\) with \(L_{k',i}>\theta\)
  \State \((w^{t+1},p^{t+1},N^{t+1})\leftarrow(w^t_{i,k'},p^t_{i,k'},N^t_{i,k'})\)
  \State \(t\leftarrow t+1\)
\EndWhile
\State \(\widehat{\mathcal{K}}\leftarrow w^t,\;\widehat{N}\leftarrow N^t,\;\widehat{NC}\leftarrow NC\)
\end{algorithmic}
\end{algorithm}

To determine which candidate update to accept, we compute for each hash \(k\) a \emph{likelihood score} \(L_{k,i} = \sum_{l=1}^i (2\,c_{k,l} - C_l)\,C_l\), where \(c_{k,l}\) is the number of nodes in YODA round \(l\) that produced hash \(k\), and \(C_l\) is the size of the execution set in that round. This score increases rapidly when many independent nodes agree, providing a strong statistical signal. We then compare \(L_{k,i}\) against a threshold \(\theta = \left(\ln\tfrac{1-\beta}{\beta}\right)\, \frac{2\,q(1-q)\,C\,(1-f_{\max})\,f_{\max}}{(1 - f_{\max}) - f_{\max}}\), which is derived to balance speed (i.e., minimizing the number of rounds) with security (i.e., filtering out Byzantine proposals). Here, \((1 - \beta)\) denotes the desired statistical confidence level, \(C\) is the total number of \cone\ nodes, \(q\) is the fraction of \cone\ nodes selected per round, and \(f_{\max}\) is the maximum allowed fraction of Byzantine nodes.


If no \(L_{k,i}>\theta\), we repeat with a larger $ES$. However on the other hands, once some \(k'\) satisfies \(L_{k',i}>\theta\), we accept and adopt its corresponding update:$(w^{t+1},\,p^{t+1},\,N^{t+1}) 
  \;=\;(w^t_{i,k'},\,p^t_{i,k'},\,N^t_{i,k'})$.


Finally, when the model performance \(M(w^t) \ge \tau\) or the maximum round limit \(t = t_{\max}\) is reached, we output the trained model \(\widehat{\mathcal{K}} = w^t\), the seller contribution vector \(\widehat{N} = N^t\), and the \cone\ node participation log \(\widehat{NC} = NC\), where each entry \(\langle id, \text{count} \rangle\) records how many times a \cone\ node with identifier \(id\) participated during training. 

These values feed back to the on‐chain contract (to method $onComputeComplete$ in Algorithm~\ref{alg:auction}) for precise, fair distribution of the payment \(Amt\).


Note that \prot\ Core is not fully off-chain. It interacts with the off-chain Corrected OSMD algorithm (Algorithm~\ref{alg:osmd}). Once an execution set $ES_i$ is selected for round $i$, an event is emitted, which is picked up by the participating \cone\ nodes. These nodes then execute Corrected OSMD off-chain. After all nodes submit their hash commitments of the computed weights-or once a block-based timeout is reached-another event is triggered to compute a likelihood score over all committed hashes in round $i$. If this score exceeds a defined threshold ($\theta$), the system proceeds to evaluate the submitted weights against the specified metrics. Otherwise, a new YODA sortition round is initiated, repeating the process. 

Next, we discuss the Corrected OSMD, which-as previously mentioned-combines the strengths of Corrected KRUM and OSMD. This module is invoked by the \prot\ Core described earlier.

\subsection{Corrected OSMD}
\label{sec:osmd}
Corrected OSMD is an algorithm designed to improve a machine learning model by intelligently selecting and aggregating updates from multiple data sellers. We know that all sellers are not equally useful and some might send malicious updates and we also want to be fair and efficient in choosing from whom we want updates. So we use Corrected OSMD as strategy to learn which sellers are most helpful over time.

Corrected OSMD is invoked by the \prot\ Core in any round \(t\) with two primary inputs: (1) the current model parameters (weights) \(w^{t}\) to be updated, and (2) a probability distribution \(p^t\), which determines the likelihood of selecting each seller. In addition, the algorithm requires several settings: \(K\), the number of sellers to sample in this round; \(\eta\), the learning rate that controls how quickly the model updates; \(\gamma_t\), the step size that limits the magnitude of each update; and \(\alpha\), a fairness parameter ensuring every seller has at least a minimal chance of being selected.

In each round, we randomly select $K$ sellers based on the probabilities in $p^t$ (line 3 in Algorithm~\ref{alg:osmd}). These sellers receive the current model and return their proposed updates, based on their own private data. We record which sellers were selected and update their access counts in $N^{t+1}$. For each seller, we then calculate how helpful their update was—the algorithm computes the gradient-based vector $\widehat{u}^t$, where each component $\widehat{u}_i^t$ is defined as:
\[
\widehat{u}_i^t = \frac{\sum_{k=0}^{K-1} \mathds{1}\{i^{t,k} = i\}}{K p_i^t} \times \big(U(w^t + \gamma_t g_i^t) - U(w^t)\big),
\]
where $\mathds{1}\{i^{t,k} = i\}$ is an indicator function that counts the occurrences of the seller $i$ in the sampled batch. The term $U(w^t + \gamma_t g_i^t) - U(w^t)$ represents the utility difference due to the model update. The scaling factor accounts for the sampling probability $p_i^t$ and ensures unbiased estimation (Line 8 in Algorithm~\ref{alg:osmd}).

Next, using these utility scores, we update the sampling distribution $p^t$. To update $p^t$ we solve the following optimization problem: \[
       p^{t+1} \;=\;
       \arg\min_{q\in\mathcal{A}}
         \bigl\langle q,\widehat u^t\bigr\rangle
         +D_\Phi(q\|p^t)
     \]
This formulation balances exploitation (preferring sellers with low estimated utility values, which indicate high contributions) and exploration (not deviating too far from the previous distribution 
$p^t$). The term $\bigl\langle q,\widehat u^t\bigr\rangle$ encourages selecting more helpful sellers, while the Bregman divergence $D_\Phi(q\|p^t)$ regularizes the change in distribution to maintain stability and fairness.

The update is computed using the Online Mirror Descent (OMD) solver. For a detailed derivation and theoretical guarantees of OMD solver, we refer the reader to the original paper \cite{ZhaoOSMD}.

Finally instead of simply averaging the updates (which could be affected by malicious or noisy inputs) we use a robust method called Corrected KRUM~\cite{DuBlockchainDML} which selects the single update closest to the mean among the refined subset, mitigating the influence of Byzantine sellers and produce the updated model parameters $w^{t+1}$. This ensured that only the most reliable updates are used, leading to a more accurate and secure learning process over time.


\begin{algorithm}[tbp]
\caption{Corrected OSMD Sampler and Aggregator}
\label{alg:osmd}
\begin{algorithmic}[1]
   \State \textbf{Input:}
      \begin{itemize}
   \item round $t$, batch size $K$
       \item OSMD learning rate $\eta$ 
       \item optimization stepsizes $\{\gamma_t\}_{t=0}^{T-1}$
       \item parameter $\alpha \in [0,1]$, $\mathcal{A} = \mathcal{P}_{n-1} \cap [\alpha/n, \infty)^n$
       \item the sampling distribution $p^t$ and model parameter $w^t$.
       \item the number of accesses for each dataset $N^t$
   \end{itemize}
   \State \textbf{Output:} 
   \begin{itemize}
       \item trained model's parameters $w^{t+1}$
       \item number of accesses for each seller $N^{t+1}(i) \ \forall \ i \in [n]$
       \item sampling distribution $p^{t+1}$
   \end{itemize}
   \State Sample \(K\) sellers \(S^t=\{i^{t,0},\dots,i^{t,K-1}\}\) with replacement from \(p^t\).
   \ForAll{\(i\in S^t\)} 
     \State Send \(w^t\) to seller \(i\); receive update \(g_i^t=\mathcal{O}_i(w^t)\).
   \EndFor
   \State Update counts: 
     \(N^{t+1}(i)=N^t(i)+\sum_k\mathbf{1}\{i^{t,k}=i\}\).
   \State Compute utility estimates
     \[
       \widehat u_i^t
       =\frac{\sum_k\mathds{1}\{i^{t,k}=i\}}{K\,p_i^t}
        \bigl[\,U(w^t+\gamma_t\,g_i^t)-U(w^t)\bigr].
     \]
   \State Solve for new sampling distribution
     \[
       p^{t+1} \;=\;
       \arg\min_{q\in\mathcal{A}}
         \bigl\langle q,\widehat u^t\bigr\rangle
         +D_\Phi(q\|p^t)
     \]
     via the OMD solver \cite{ZhaoOSMD}.
   \State Aggregate model updates with Corrected Krum \cite{DuBlockchainDML}:
     \[
       w^{t+1}
       \;=\;\text{Krum}\bigl\{\,w^t+\gamma_t\,g_i^t\,:\,i\in S^t\}\,.
     \]
\end{algorithmic}
\end{algorithm}


\section{Game theoretic analysis of \prot}
\label{sec:game}



To ensure the security of \prot\ against Byzantine \cone\ nodes, we rely on the security guarantees provided by YODA. According to YODA’s analysis, the probability that malicious weights are accepted is upper bounded by
\[
\beta = \frac{1}{1 + \exp\!\left(
    \frac{\theta(1 - 2f_{\max})}{
        2q(1 - q)M(1 - f_{\max})f_{\max}
    }
\right)}.
\]
All parameters in the above expression have been defined previously in Section~\ref{sec:D2Mcore}. By choosing a sufficiently large value of $\theta$, the probability of malicious success, $\beta$, can be made negligibly small. 

However, since the entities participating in the \prot\ may be rational as their actions are driven by corresponding payoffs, a game-theoretic security analysis of \prot\ becomes essential. This section focuses exclusively on such an analysis.

We model the system as a strategic game with three players: the Buyer ($B$), the Seller ($S$), and the CONE node ($C$). The underlying blockchain is assumed to be secure against all forms of adversarial behavior; hence, it is not explicitly modeled as a player.

The Buyer has two available strategies:  
(1) behave honestly by placing a valid bid for an auction ($B_H$), or  
(2) behave maliciously by spamming the system with multiple bids to dominate the auction ($B_A$).  Each bid incurs a transaction fee $f$, making spamming economically inefficient. Consequently, the Buyer maximizes his utility by adopting the honest strategy $B_H$. Furthermore, as high-quality data ultimately benefits the Buyer, collusion with \cone\ nodes and/or Sellers provides no additional advantage and is therefore excluded from the game model.

Having said that, in this section we show conditions under which (i) the \emph{Seller} finds it a dominant strategy to submit honest/high-quality data, and (ii) an individual \emph{CONE node} finds it a dominant strategy to validate honestly. 


\subsection*{1. Game strategies and states}


The seller has two strategies: 
(1) Submit high-quality data for training (\(S_H\)), or 
(2) Submit low-quality data (\(S_L\)), possibly colluding with a CONE node via a bribe \(b\) for considering his weight irrespective of the data quality. The seller's payout is proportional to the quality \(q \in [0,1]\).

The CONE node has two strategies: 
(1) Validate honestly (\(C_H\)), or 
(2) Collude with the seller (\(C_C\)) by accepting bribe $b$. 




Because either the 
\cone\ node or the seller being malicious would not yield higher payoffs than behaving honestly, rational participants have no incentive to deviate from honest behavior. For instance, if all \cone\ nodes act honestly, any attempt by the seller to bribe them results in a loss for the seller. Conversely, if a 
\cone\ node behaves maliciously, it gains nothing by accepting or approving data of low quality. Therefore, in this game, the protocol can exist only in two  states: (i) both are honest, or (ii) both are malicious. The details of the protocol state for one-shot interaction is as follows:

\begin{itemize}
  \item State Honest (H): All the entities, i.e., Buyer, Seller, and \cone\ nodes are honest. This means that likelihood score of honest digest reach the threshold ($\theta$) in first round itself, with the execution set size $n_0$. 
  
  \item State Collude-Success (CS): Seller plays $S_L$, offers bribe, and all (or required coalition of) nodes accept the malicious weights immediately; collusion succeeds with probability $\beta$ (function of system design).
\end{itemize}

To summarize above, the game start with the initial state, say $I_0$, and based on the action from \cone\ nodes and sellers, it reaches either state $H$ or state $CS$. Next we will prove that it is dominant strategy to behave honestly for \cone\ nodes and sellers. For this we assume, Seller total payout when weights are accepted is $W$ and the seller's share is proportional to quality $q\in[0,1]$ (honest $q$ high, malicious $q$ low). Collusion may change effective weight to $\tilde q>q$. Furthermore, A collusion attempt has probability $\beta$ of immediate success and the \cone\ node's effective earning $E$, after they perform the computation to increase the likelihood score of weights (honest or malicious), where total number of \cone\ nodes are $N$.  


\subsection*{2. Payoff functions (one-shot expected payoff)}

\paragraph{Seller.} Let the share of bid amount that is allocated to $A_s$. Each seller gets a share of $A_s$ based on the quality of data they provide that matches the tag passed with bid. So the effective earning of a seller when behaved honestly, is $u_s(S_H)=qA_n$. However, as discussed earlier seller may behave maliciously by bribing a \cone\ nodes and falsify the quality to $q'$, where $q'>q$. IN this case the effective payoff for the seller is $u_s(S_L) = (1-\beta)qA_n + \beta(q'A_n-Nb)$


\paragraph{\cone\ node.}
Let the total number of CONE nodes be denoted by \( N \). Let the total reward  for CONE participation (for the whole protocol run) be \( A_c \). This entire amount is distributed across all executions and then allocated to nodes proportionally to how many times each node was selected across all rounds.

The seller offers a bribe \( b \) to a CONE node for accepting false weights. If the node tries to collude, the collusion succeeds with probability \( \beta \) (i.e., the node receives the bribe only when the cheating attempt succeeds).


The protocol runs for \( r \) rounds. When nodes behave honestly, it requires \( r_{\text{hon}} \) rounds; when a few nodes behave maliciously, it requires \( r_{\text{mal}} \geq r_{\text{hon}} \) rounds.

The per-round execution set (ES) growth rule is given by $ES_i = ES_{i-1} + 2^{\,i-1}, \quad i \ge 2$, and we denote \( ES_1 = s_0 \) (a given base size). 

We assume that each round’s execution set is chosen uniformly at random from the \( N \) nodes (via sortition), independently across rounds. Hence, the probability that a given node is selected in round \( i \) is \( \tfrac{ES_i}{N} \). This follows the standard YODA assumption of unbiased random sampling. From the recurrence with \( ES_1 = s_0 \): $ES_i = s_0 + \sum_{k=1}^{i-1} 2^{k-1} = s_0 + (2^{i-1} - 1)$.
Hence, $ES_i = s_0 + 2^{i-1} - 1$. The total number of executions (selections) across all \( r \) rounds (counting multiplicity) is
\[
S_r = \sum_{i=1}^{r} ES_i = \sum_{i=1}^{r} (s_0 + 2^{i-1} - 1)
      = r(s_0 - 1) + \sum_{i=1}^{r} 2^{i-1}.
\]
Because, $\sum_{i=1}^{r} 2^{i-1} = 2^{r} - 1$, $S_r$ is simplified to $S_r = r(s_0 - 1) + 2^{r} - 1$.

By linearity and under the uniform sampling assumption, the probability that a fixed node is selected in round \( i \) is \( \tfrac{ES_i}{N} \). Therefore, the expected total number of times a fixed node is selected across \( r \) rounds is
$\mathbb{E}[\text{selections of node}] 
= \sum_{i=1}^{r} \frac{ES_i}{N} 
= \frac{S_r}{N}$.

Given this, the expected share of \( A_c \) for a given node is therefore $\mathbb{E}[\text{reward from } A_c]
= A_c \cdot \frac{\mathbb{E}[\text{selections of node}]}{\mathbb{E}[\text{total selections}]}$. By substituting the value we have, 


\[
u_c(C_H)=\mathbb{E}[\text{reward from } A_c]
= A_c \cdot \frac{S_r / N}{S_r}
= \frac{A_c}{N}.
\]
From this we can observe that, under the uniform-random selection model, the expected share of the \( A_c \) for any fixed node is independent of the number of rounds \( r \) and independent of how fast \( ES_i \) grows. So to penalize the nodes that attempt to increase the likelihood score of the malicious weights, \prot\ doesn't pay the per-node share to the node implicated in producing the (fraudulent) digest and hence colluders risk losing the early-round payments as well. 

To represented it formally, let \( q_{\text{caught}}(r) \) be the probability that a colluding node is excluded from receiving the protocol share when termination occurs at round \( r \). Then, the expected protocol reward for the colluder becomes

The expected protocol reward for a colluding node is \( A_c (1 - q_{\text{caught}}(r)) \). Thus, the total expected payoff for colluding node is $u_c(C_C)=
= \frac{A_c}{N} (1 - q_{\text{caught}}(r)) + \beta b$. Note that \( q_{\text{caught}}(r) \) increases with \( r \) since the execution set grows exponentially, adding more honest nodes and thus more evidence for detection. Consequently, \( u_c(C_C) \) decreases with \( r \).

\subsection{Dominant strategy for players in the game}
Given the defined payoffs, we now prove that honest behavior constitutes a dominant strategy. This is established through two lemmas and a theorem: Lemma~\ref{lm:p1} demonstrates that honest behavior is the dominant strategy for the seller, Lemma~\ref{lm:p2} establishes the same for the \cone\ nodes, and Theorem~\ref{thm:p1} concludes that honesty is the dominant strategy for all players.

\begin{lemma}
    \label{lm:p1}
    For the seller, honesty is a dominant strategy if 
    $N b \ge (q' - q)A_n$.
\end{lemma}

\begin{proof}
As explained earlier, the seller's utilities when behaving honestly and maliciously are
\[
u_s(S_H)=qA_n,\qquad
u_s(S_L)=(1-\beta)qA_n+\beta\big(q' A_n - N b\big).
\]
Given this, the difference $\Delta$ is calculated as, 
\[
\Delta := u_s(S_H)-u_s(S_L)
= qA_n -\big[(1-\beta)qA_n+\beta(q'A_n - N b)\big].
\]
Simplifying the above gives
\[
\Delta
= qA_n - qA_n + \beta qA_n - \beta q' A_n + \beta N b
= \beta\big(N b - (q'-q)A_n\big).
\]
Since $\beta$ is the negligible but non zero probability, for honest strategy to dominate $\Delta\ge 0$ must hold, which leads to $N b \;\ge\; (q'-q)A_n$. According to our assumption, this condition holds.

It is important to note that \( N \) may represent either the total number of \cone\ nodes or the subset of \cone\ nodes selected in the execution sets, depending on how the seller implements bribery. Regardless of the seller’s bribery strategy, the proof remains valid. Moreover, the condition \( N b \ge (q' - q)A_n \) can be satisfied by increasing the total number of \cone\ nodes and/or the size of the execution set in Algorithm~\ref{alg:workflow}.


\end{proof}

\begin{lemma}
    \label{lm:p2}
    Then, honest behaviour $C_H$ is a (weakly) dominant strategy for the \cone\ node if $q_{\text{caught}}(r) \cdot \frac{A_c}{N} \;\ge\; \beta b$.
\end{lemma}
\begin{proof}
As discussed earlier we know that, 
\[
u_c(C_H) = \frac{A_c}{N}, 
\qquad
u_c(C_C) = \frac{A_c}{N}\big(1 - q_{\text{caught}}(r)\big) + \beta b.
\]
The difference in the honest and malicious utility is:
\[
\Delta u_c=u_c(C_H) - u_c(C_C)
= \frac{A_c}{N} - \left[ \frac{A_c}{N}\big(1 - q_{\text{caught}}(r)\big) + \beta b \right].
\] 
which gives, 
\[
\Delta u_c
= q_{\text{caught}}(r) \cdot \frac{A_c}{N} - \beta b.
\]
For honest behavior to be dominant, the utility difference \( \Delta u_c \) must satisfy \( \Delta u_c > 0 \), which holds under our assumption. This assumption is justified since \( q_{\text{caught}}(r) \) increases with the number of rounds, and under malicious behavior, YODA requires more rounds to terminate compared to the honest case. However, the expected gain for a honest \cone\ nodes is independent on the number of rounds.
\end{proof}




\begin{theorem}
    \label{thm:p1}
Then honest behavior is a dominant strategy for every player (the seller and each CONE node) if $\beta (q'-q) A_n \le q_{\mathrm{caught}}(r) A_c$.

\end{theorem}

\begin{proof}
From Lemma~\ref{lm:p1}, honesty is dominant for the seller provided $N b \ge (q' - q)A_n$, while from Lemma~\ref{lm:p2} honesty is dominant for a CONE node provided $q_{\mathrm{caught}}(r)\cdot\frac{A_c}{N} \ge \beta b$.

These two inequalities impose a lower and an upper bound on the feasible per-node bribe \(b\).
A value \(b\) satisfying both bounds exists iff the lower bound does not exceed the upper bound. 
\[
\frac{(q' - q)A_n}{N} \le \frac{q_{\mathrm{caught}}(r)A_c}{\beta N},
\] 
which gives 
\[
\beta (q' - q) A_n \le q_{\mathrm{caught}}(r) A_c.
\]
According to our assumption this condition holds which states that there exists a \(b\) that simultaneously satisfies the seller's and the CONE node's constraints~\footnote{Edge case: if \(\beta=0\) then the seller's and node's payoffs coincide irrespective of \(b\) and the players are at least weakly indifferent; the inequality above becomes \(0\le q_{\mathrm{caught}}(r)A_c\), which holds trivially.}. For any such \(b\) (which \prot\ can enforce), neither the seller nor any CONE node can increase their payoff by unilaterally deviating to malicious behavior, so honesty is a dominant strategy for all players.

\end{proof}

\section{\prot \  Evaluation}
\label{sec:evaluation}
In this section, we will present the evaluation of \prot \  along with the findings from this evaluation. Before diving into the results, let’s first discuss the experimental setup.

\subsection{Experimental Setup}
The experiment is conducted on a 3 high-performance machine configured with dual AMD EPYC 7452 32-core processors, providing a total of 64 physical cores and 128- threads (with 4 threads per core). We run 100-250 users on this machine (50-200 sellers and 50 buyers) and a Ethereum blockchain network consists of 50 nodes\footnote{This is because the top 50 miners 
(by mining power) contribute to around 99.98\% of total mining 
power of the real Ethereum network, with the most powerful miner 
controlling ${\sim}33\%$ of the total mining power.}, which we setup using ethereum-package of kurtosis~\cite{ethkurtosis}


For computational tasks (\cone\ nodes), the system utilizes a high-performance 24GB NVIDIA RTX A5000 GPU running on Ubuntu 20.04, enabling efficient parallel processing—especially for machine learning and model training. A \cone\ node network comprising 50 nodes has been deployed on this machine.


Furthermore, we have implemented \prot\ using Python3.0. and PyTorch~\cite{paszke2019pytorchimperativestylehighperformance}, an open-source machine learning framework 
The details of the model, the number of nodes that contribute to each \prot\ components, metrics used for the evaluation, and the datasets we used are summarized in Table \ref{tab:setup}.


Real-world federated learning systems often operate under non-IID settings due to user behavior or geographic differences. To mimic this, we construct non-IID client datasets from MNIST, FashionMNIST and CIFAR-10 by allocating data with imbalanced label distributions using a Dirichlet-based partitioning strategy(\(\alpha=0.5\))~\cite{li2021federatedlearningnoniiddata}. Each dataset is split into three distinct subsets: a training set, a validation set, and a test set. The validation and test datasets are kept separate to evaluate the model’s generalization ability during and after training. 


\begin{table}[H]
\centering
\begin{tabular}{l|l}
\toprule
\textbf{Parameter} & \textbf{Value} \\
\midrule
Model Architecture & 1 Convolution layer Neural Network \\
Datasets & CIFAR10, MNIST, FashionMNIST \\
Optimization & Adam \\
Learning Rate & 0.01 \\
Batch Size & 64 \\
Number of Sellers & 50-200 \\
Number of Arbiters & 50 \\
Number of Rounds & 50 \\
Epochs per Round & 3 \\
\bottomrule
\end{tabular}
\caption{Experimental setup specifications used for evaluating \prot \ }
\label{tab:setup}
\end{table}

\subsection{Parameters and Metrics}
\label{sec:paraAndMetrics}
To thoroughly assess the performance of our proposed system, we conducted a series of experiments using a range of parameters and evaluation metrics.

We evaluated the system’s accuracy across rounds by varying the number of rounds from 1 to 50 on three datasets. While precision and recall were also measured, they exhibited trends similar to accuracy and showed negligible variation; therefore, only accuracy results are reported. To understand convergence, we determined the saturation point, defined as the minimum number of rounds required for accuracy to reach its maximum achievable value, beyond which additional rounds yield negligible improvement. We also assessed the system’s scalability by measuring the time required to complete 50 rounds for different numbers of sellers, specifically 50, 100, 150, and 200, for each dataset, providing insights into computational efficiency as network size grows.

To evaluate Byzantine resilience, we simulated scenarios where 20\%, 30\%, 40\%, and 50\% of \cone\ nodes acted maliciously by returning random or stale updates. Under these conditions, both accuracy and completion time were measured to assess the system’s fault tolerance and efficiency. Additionally, ablation studies were conducted to analyze the contribution of key components: (i) No Krum Aggregation, where Corrected OSMD aggregation was disabled to highlight vulnerability to model poisoning, and (ii) No Yoda Consensus, where the consensus mechanism was disabled while retaining Corrected OSMD aggregation, allowing us to study model divergence and consistency across nodes under adversarial scenarios.

With this setup, we performed the experiments for a total of four times and the averaged results of which are elaborated in the next section (section \ref{sec:results}).

\subsection{Evaluation Result}
\label{sec:results}


\pgfplotsset{small,label style={font=\fontsize{8}{9}\selectfont},legend style={font=\fontsize{4.5}{8}\selectfont},height=3.8cm,width=0.35\textwidth}
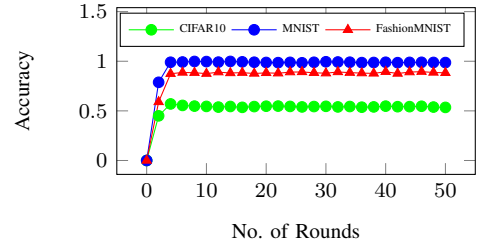
\begin{figure}
\centering

\begin{tikzpicture}
\begin{axis}[
    enlargelimits=0.10,
    legend columns=3,
    legend pos=north east,
    xlabel={No. of Rounds},
    ylabel= {Accuracy},
    grid=minor,
    xtick={0,10,20,30,40,50},
    xmax=50,
    ymin=0,
    ymax=1.4,
    ]
    
    \addplot[green, mark=*] table [x=Rounds, y=CIFAR10_Accuracy, col sep=comma] {data/accuracy_data.csv};
    \addplot[blue, mark=*] table [x=Rounds, y=MNIST_Accuracy, col sep=comma] {data/accuracy_data.csv}; 
    \addplot[red, mark=triangle*] table [x=Rounds, y=FashionMNIST_Accuracy, col sep=comma] {data/accuracy_data.csv};


    \addlegendentry{CIFAR10}
    \addlegendentry{MNIST}
    \addlegendentry{FashionMNIST}
\end{axis}
\end{tikzpicture}

\caption{Variation in the accuracy over number of rounds for different datasets.} 
\label{fig:metrics}
\end{figure}

We evaluated the system’s accuracy over varying rounds across three datasets. For MNIST and FashionMNIST, \prot\ reached over 90\% and 80\% accuracy, respectively, within just five rounds, after which performance stabilized. In contrast, CIFAR-10 accuracy plateaued around 60\%, as shown in Figure \ref{fig:metrics}. These results indicate that \prot\ rapidly converges to its optimal performance within a 5 rounds. The relatively lower  saturation on CIFAR-10 can be attributed to the stronger effect of non-IID data distributions, where higher class diversity across nodes impacts convergence.



To assess the scalability of the system, we measured the time required to complete 20 rounds for varying numbers of sellers, where we observed that the time taken by \prot\ increases almost linearly with the number of sellers in the system, demonstrating strong scalability, as shown in Figure \ref{fig:time}. The reported time corresponds to training the model for 20 rounds. However, as previously observed, saturation is typically achieved by round 5. Therefore, the actual required training time is approximately four times less than what is shown in the plot.


\pgfplotsset{small,label style={font=\fontsize{8}{9}\selectfont},legend style={font=\fontsize{4.5}{8}\selectfont},height=3.8cm,width=0.35\textwidth}
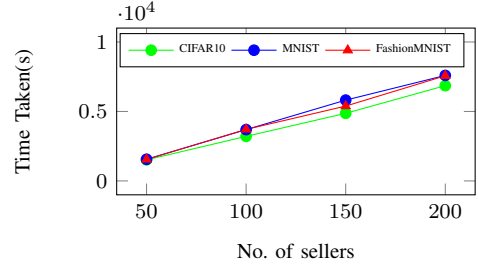
\begin{figure}
\centering

\begin{tikzpicture}
\begin{axis}[
    enlargelimits=0.10,
    legend columns=3,
    legend pos=north east,
    xlabel={No. of sellers},
    ylabel= {Time Taken(s)},
    grid=minor,
    xtick={50,100,150,200},
    xmax=200,
    ymin=0,
    ymax=10000,
    ]
    
    \addplot[green, mark=*] table [x=No_of_Sellers, y=CIFAR10_Time, col sep=comma] {data/time_vs_sellers.csv};
    \addplot[blue, mark=*] table [x=No_of_Sellers, y=MNIST_Time, col sep=comma] {data/time_vs_sellers.csv}; 
    \addplot[red, mark=triangle*] table [x=No_of_Sellers, y=FashionMNIST_Time, col sep=comma] {data/time_vs_sellers.csv};
    \addlegendentry{CIFAR10}
    \addlegendentry{MNIST}
    \addlegendentry{FashionMNIST}

\end{axis}
\end{tikzpicture}

\caption{Time taken to run 20 rounds of the protocol on different datasets, measured against the number of sellers.} 
\label{fig:time}
\end{figure}

\pgfplotsset{small,label style={font=\fontsize{8}{9}\selectfont},legend style={font=\fontsize{4.5}{8}\selectfont},height=3.8cm,width=0.35\textwidth}
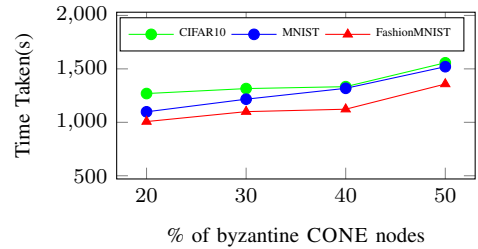
\begin{figure}
\centering

\begin{tikzpicture}
\begin{axis}[
    enlargelimits=0.10,
    legend columns=3,
    legend pos=north east,
    xlabel={\% of byzantine \cone\ nodes},
    ylabel= {Time Taken(s)},
    grid=minor,
    xtick={20,30,40,50},
    xmax=50,
    ymin=600,
    ymax=1900,
    ]
    
    \addplot[green, mark=*] table [x=Byzantine, y=CIFAR10_Time, col sep=comma] {data/time_round_byzantine.csv};
    \addplot[blue, mark=*] table [x=Byzantine, y=MNIST_Time, col sep=comma] {data/time_round_byzantine.csv}; 
    \addplot[red, mark=triangle*] table [x=Byzantine, y=FashionMNIST_Time, col sep=comma] {data/time_round_byzantine.csv};
    \addlegendentry{CIFAR10}
    \addlegendentry{MNIST}
    \addlegendentry{FashionMNIST}

\end{axis}
\end{tikzpicture}

\caption{Time taken to run the protocol for 20 rounds on different datasets in the presence of different percentage of byzantine nodes.} 
\label{fig:byz_time}
\end{figure}


Furthermore, we analyzed the system's runtime across rounds 1 to 20, excluding the model training time, under varying proportions of Byzantine nodes. As shown in Figure~\ref{fig:byz_time}, the protocol continues to demonstrate near-linear scalability. Importantly, the total time taken by the system does not grow exponentially with an increasing number of Byzantine nodes, indicating the protocol’s robustness and efficiency under adversarial conditions-despite a slightly steeper slope when the percentage of adversarial \cone\ nodes increases.

In the Byzantine setup, \prot\ maintains stable accuracy and performance comparable to the honest setting when the fraction of Byzantine nodes is up to 30\%, as shown in Figures~\ref{fig:CFAR10_byzantine_prot}, \ref{fig:MNIST_byzantine_prot}, and \ref{fig:FashionMNIST_byzantine_prot}. Beyond this threshold, accuracy begins to decline, with the effect being most pronounced on the CIFAR-10 dataset. For MNIST and FashionMNIST, the degradation is comparatively smaller, though noticeable fluctuations persist, suggesting some instability in the system. This behavior can be attributed to the fact that YODA requires a greater number of rounds to counter higher levels of adversarial participation. The sharper decline on CIFAR-10, relative to MNIST and FashionMNIST, can further be explained by the increased complexity of the dataset.


To isolate the contribution of YODA and Corrected OSMD in \prot, we ran ablation studies removing each component. Without Corrected OSMD, CIFAR-10 accuracy degrades sharply as adversarial participation grows: at 20\% adversaries it stabilizes around 50--55\%, but beyond 30\% it oscillates between 10\% and 55\% (Figure~\ref{fig:CFAR10_byzantine_without_crum}). Fashion-MNIST is more stable, maintaining 88--90\% at 20\% and mostly above 80\% at 30\%, but plunging below 20\% at 40--50\% adversaries (Figure~\ref{fig:FashionMNIST_byzantine_without_crum}). MNIST shows the highest robustness, remaining near-perfect at 20--30\% adversaries with fluctuations appearing only at 40--50\% (Figure~\ref{fig:MNIST_byzantine_without_crum}). Overall, performance is stable up to 20\% adversaries but becomes highly unstable beyond this, as random sampling without Corrected OSMD can select poor-quality weights in some rounds, causing abrupt drops.


\pgfplotsset{small,label style={font=\fontsize{8}{9}\selectfont},legend style={font=\fontsize{4.5}{8}\selectfont},height=3.6cm,width=\textwidth}
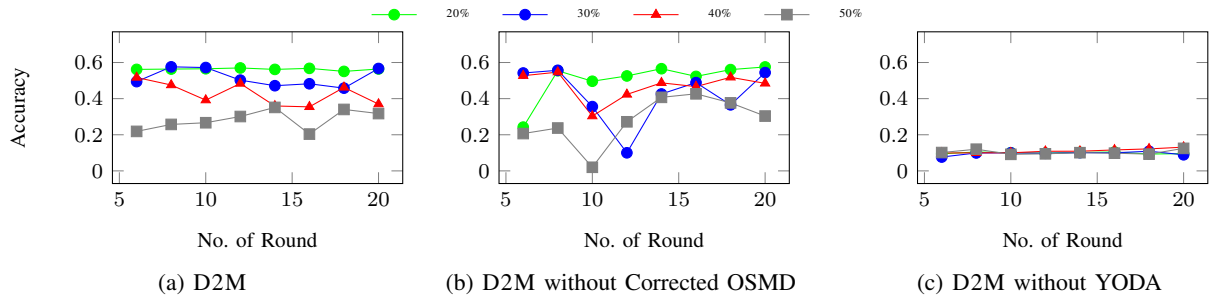
\begin{figure*}
\centering

\begin{tikzpicture}
\begin{customlegend}[legend columns=5,legend style={align=left,draw=none,column sep=2ex},
        legend entries={
                        \textsc{20\%} ,
                        \textsc{30\%} ,
                        \textsc{40\%} ,
                        \textsc{50\%} ,
                        }]
        \addlegendimage{green, mark=*}
        \addlegendimage{blue, mark=*}
        \addlegendimage{red, mark=triangle*}
        \addlegendimage{gray, mark=square*}
        
\end{customlegend}
\end{tikzpicture}

\begin{subfigure}{0.30\linewidth}
\centering
\begin{tikzpicture}
\begin{axis}[
    enlargelimits=0.10,
    legend columns=3,
    legend pos=north east,
    xlabel={No. of Round},
    ylabel= {Accuracy},
    grid=minor,
    xmax=20,
    ymin=0,
    ymax=0.7,
    ]
    
    \addplot[green, mark=*] table [x=Communication_Round_1, y=20_1, col sep=comma] {data/d2m.csv};
    \addplot[blue, mark=*] table [x=Communication_Round_1, y=30_1, col sep=comma] {data/d2m.csv}; 
    \addplot[red, mark=triangle*] table [x=Communication_Round_1, y=40_1, col sep=comma] {data/d2m.csv};
    \addplot[gray, mark=square*] table [x=Communication_Round_1, y=50_1, col sep=comma] {data/d2m.csv};

\end{axis}
\end{tikzpicture}
\caption{\prot}
\label{fig:CFAR10_byzantine_prot}
\end{subfigure}
\begin{subfigure}{0.30\linewidth}
\centering
\begin{tikzpicture}
\begin{axis}[
    enlargelimits=0.10,
    legend columns=3,
    legend pos=north east,
    xlabel={No. of Round},
    grid=minor,
    xmax=20,
    ymin=0,
    ymax=0.7,
    ]
    
    \addplot[green, mark=*] table [x=Communication_Round_1, y=20_1, col sep=comma] {data/no_krum.csv};
    \addplot[blue, mark=*] table [x=Communication_Round_1, y=30_1, col sep=comma] {data/no_krum.csv}; 
    \addplot[red, mark=triangle*] table [x=Communication_Round_1, y=40_1, col sep=comma] {data/no_krum.csv};
    \addplot[gray, mark=square*] table [x=Communication_Round_1, y=50_1, col sep=comma] {data/no_krum.csv};

\end{axis}
\end{tikzpicture}
\caption{\prot\ without Corrected OSMD}
\label{fig:CFAR10_byzantine_without_crum}
\end{subfigure}
\begin{subfigure}{0.30\linewidth}
\centering
\begin{tikzpicture}
\begin{axis}[
    enlargelimits=0.10,
    legend columns=3,
    legend pos=north east,
    xlabel={No. of Round},
    grid=minor,
    xmax=20,
    ymin=0,
    ymax=0.7,
    ]
    
    \addplot[green, mark=*] table [x=Communication_Round_1, y=20_1, col sep=comma] {data/no_yoda.csv};
    \addplot[blue, mark=*] table [x=Communication_Round_1, y=30_1, col sep=comma] {data/no_yoda.csv}; 
    \addplot[red, mark=triangle*] table [x=Communication_Round_1, y=40_1, col sep=comma] {data/no_yoda.csv};
    \addplot[gray, mark=square*] table [x=Communication_Round_1, y=50_1, col sep=comma] {data/no_yoda.csv};

\end{axis}
\end{tikzpicture}
\caption{\prot\ without YODA}
\label{fig:CFAR10_byzantine_without_yoda}
\end{subfigure}

\caption{Time taken to run the protocol for 20 rounds on CFAR10 in the presence of different percentage of byzantine nodes.} 
\label{fig:CFAR10_byzantine}
\end{figure*}

\pgfplotsset{small,label style={font=\fontsize{8}{9}\selectfont},legend style={font=\fontsize{4.5}{8}\selectfont},height=3.6cm,width=\textwidth}
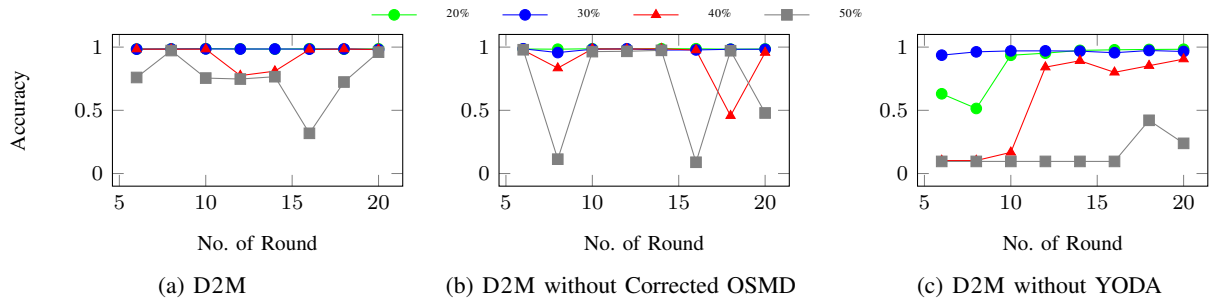
\begin{figure*}
\centering

\begin{tikzpicture}
\begin{customlegend}[legend columns=5,legend style={align=left,draw=none,column sep=2ex},
        legend entries={
                        \textsc{20\%} ,
                        \textsc{30\%} ,
                        \textsc{40\%} ,
                        \textsc{50\%} ,
                        }]
        \addlegendimage{green, mark=*}
        \addlegendimage{blue, mark=*}
        \addlegendimage{red, mark=triangle*}
        \addlegendimage{gray, mark=square*}
        
\end{customlegend}
\end{tikzpicture}

\begin{subfigure}{0.30\linewidth}
\centering
\begin{tikzpicture}
\begin{axis}[
    enlargelimits=0.10,
    legend columns=3,
    legend pos=north east,
    xlabel={No. of Round},
    ylabel= {Accuracy},
    grid=minor,
    xmax=20,
    ymin=0,
    ymax=1,
    ]
    
    \addplot[green, mark=*] table [x=Communication_Round_1, y=20_1, col sep=comma] {data/d2m_MNIST.csv};
    \addplot[blue, mark=*] table [x=Communication_Round_1, y=30_1, col sep=comma] {data/d2m_MNIST.csv}; 
    \addplot[red, mark=triangle*] table [x=Communication_Round_1, y=40_1, col sep=comma] {data/d2m_MNIST.csv};
    \addplot[gray, mark=square*] table [x=Communication_Round_1, y=50_1, col sep=comma] {data/d2m_MNIST.csv};

\end{axis}
\end{tikzpicture}
\caption{\prot}
\label{fig:MNIST_byzantine_prot}
\end{subfigure}
\begin{subfigure}{0.30\linewidth}
\centering
\begin{tikzpicture}
\begin{axis}[
    enlargelimits=0.10,
    legend columns=3,
    legend pos=north east,
    xlabel={No. of Round},
    grid=minor,
    xmax=20,
    ymin=0,
    ymax=1,
    ]
    
    \addplot[green, mark=*] table [x=Communication_Round_1, y=20_1, col sep=comma] {data/no_krum_MNIST.csv};
    \addplot[blue, mark=*] table [x=Communication_Round_1, y=30_1, col sep=comma] {data/no_krum_MNIST.csv}; 
    \addplot[red, mark=triangle*] table [x=Communication_Round_1, y=40_1, col sep=comma] {data/no_krum_MNIST.csv};
    \addplot[gray, mark=square*] table [x=Communication_Round_1, y=50_1, col sep=comma] {data/no_krum_MNIST.csv};

\end{axis}
\end{tikzpicture}
\caption{\prot\ without Corrected OSMD}
\label{fig:MNIST_byzantine_without_crum}
\end{subfigure}
\begin{subfigure}{0.30\linewidth}
\centering
\begin{tikzpicture}
\begin{axis}[
    enlargelimits=0.10,
    legend columns=3,
    legend pos=north east,
    xlabel={No. of Round},
    grid=minor,
    xmax=20,
    ymin=0,
    ymax=1,
    ]
    
    \addplot[green, mark=*] table [x=Communication_Round_1, y=20_1, col sep=comma] {data/no_yoda_MNIST.csv};
    \addplot[blue, mark=*] table [x=Communication_Round_1, y=30_1, col sep=comma] {data/no_yoda_MNIST.csv}; 
    \addplot[red, mark=triangle*] table [x=Communication_Round_1, y=40_1, col sep=comma] {data/no_yoda_MNIST.csv};
    \addplot[gray, mark=square*] table [x=Communication_Round_1, y=50_1, col sep=comma] {data/no_yoda_MNIST.csv};

\end{axis}
\end{tikzpicture}
\caption{\prot\ without YODA}
\label{fig:MNIST_byzantine_without_yoda}
\end{subfigure}

\caption{Time taken to run the protocol for 20 rounds on MNIST in the presence of different percentage of byzantine nodes.} 
\label{fig:MNIST_byzantine}
\end{figure*}

\pgfplotsset{small,label style={font=\fontsize{8}{9}\selectfont},legend style={font=\fontsize{4.5}{8}\selectfont},height=3.6cm,width=\textwidth}
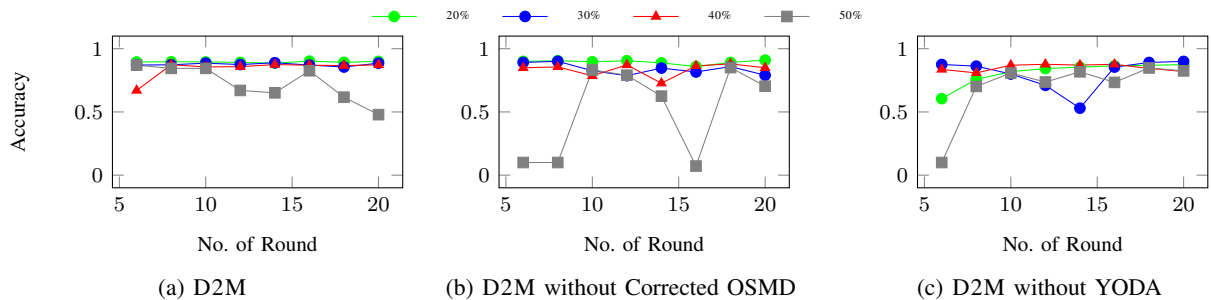
\begin{figure*}
\centering

\begin{tikzpicture}
\begin{customlegend}[legend columns=5,legend style={align=left,draw=none,column sep=2ex},
        legend entries={
                        \textsc{20\%} ,
                        \textsc{30\%} ,
                        \textsc{40\%} ,
                        \textsc{50\%} ,
                        }]
        \addlegendimage{green, mark=*}
        \addlegendimage{blue, mark=*}
        \addlegendimage{red, mark=triangle*}
        \addlegendimage{gray, mark=square*}
        
\end{customlegend}
\end{tikzpicture}

\begin{subfigure}{0.30\linewidth}
\centering
\begin{tikzpicture}
\begin{axis}[
    enlargelimits=0.10,
    legend columns=3,
    legend pos=north east,
    xlabel={No. of Round},
    ylabel= {Accuracy},
    grid=minor,
    xmax=20,
    ymin=0,
    ymax=1,
    ]
    
    \addplot[green, mark=*] table [x=Communication_Round_1, y=20_1, col sep=comma] {data/d2m_FashionMNIST.csv};
    \addplot[blue, mark=*] table [x=Communication_Round_1, y=30_1, col sep=comma] {data/d2m_FashionMNIST.csv}; 
    \addplot[red, mark=triangle*] table [x=Communication_Round_1, y=40_1, col sep=comma] {data/d2m_FashionMNIST.csv};
    \addplot[gray, mark=square*] table [x=Communication_Round_1, y=50_1, col sep=comma] {data/d2m_FashionMNIST.csv};

\end{axis}
\end{tikzpicture}
\caption{\prot}
\label{fig:FashionMNIST_byzantine_prot}
\end{subfigure}
\begin{subfigure}{0.30\linewidth}
\centering
\begin{tikzpicture}
\begin{axis}[
    enlargelimits=0.10,
    legend columns=3,
    legend pos=north east,
    xlabel={No. of Round},
    grid=minor,
    xmax=20,
    ymin=0,
    ymax=1,
    ]
    
    \addplot[green, mark=*] table [x=Communication_Round_1, y=20_1, col sep=comma] {data/no_krum_FashionMNIST.csv};
    \addplot[blue, mark=*] table [x=Communication_Round_1, y=30_1, col sep=comma] {data/no_krum_FashionMNIST.csv}; 
    \addplot[red, mark=triangle*] table [x=Communication_Round_1, y=40_1, col sep=comma] {data/no_krum_FashionMNIST.csv};
    \addplot[gray, mark=square*] table [x=Communication_Round_1, y=50_1, col sep=comma] {data/no_krum_FashionMNIST.csv};

\end{axis}
\end{tikzpicture}
\caption{\prot\ without Corrected OSMD}
\label{fig:FashionMNIST_byzantine_without_crum}
\end{subfigure}
\begin{subfigure}{0.30\linewidth}
\centering
\begin{tikzpicture}
\begin{axis}[
    enlargelimits=0.10,
    legend columns=3,
    legend pos=north east,
    xlabel={No. of Round},
    grid=minor,
    xmax=20,
    ymin=0,
    ymax=1,
    ]
    
    \addplot[green, mark=*] table [x=Communication_Round_1, y=20_1, col sep=comma] {data/no_yoda_FashionMNIST.csv};
    \addplot[blue, mark=*] table [x=Communication_Round_1, y=30_1, col sep=comma] {data/no_yoda_FashionMNIST.csv}; 
    \addplot[red, mark=triangle*] table [x=Communication_Round_1, y=40_1, col sep=comma] {data/no_yoda_FashionMNIST.csv};
    \addplot[gray, mark=square*] table [x=Communication_Round_1, y=50_1, col sep=comma] {data/no_yoda_FashionMNIST.csv};

\end{axis}
\end{tikzpicture}
\caption{\prot\ without YODA}
\label{fig:FashionMNIST_byzantine_without_yoda}
\end{subfigure}

\caption{Time taken to run the protocol for 20 rounds on FashionMNIST in the presence of different percentage of byzantine nodes.} 
\label{fig:FashionMNIST_byzantine}
\end{figure*}

In the No-YODA setting—where Corrected OSMD operates without consensus verification—the model becomes highly vulnerable to poisoning. CIFAR-10 fails to converge even after 20 rounds (Figure~\ref{fig:CFAR10_byzantine_without_yoda}), while MNIST and FashionMNIST maintain reasonable accuracy only under low adversarial ratios (20--30\%) (Figures~\ref{fig:MNIST_byzantine_without_yoda}, \ref{fig:FashionMNIST_byzantine_without_yoda}). Accuracy collapses at 40--50\% adversaries, showing that Corrected OSMD alone cannot prevent consensus on malicious seller updates. On more complex datasets such as CIFAR-10, adversarial influence dominates and drives convergence toward corrupted gradients. These results highlight that integrating YODA consensus with Corrected OSMD is essential for robustness under Byzantine conditions.

\section{Related Work}
\label{sec:related}
\subsection{Centralized Data Marketplaces}
Commercial offerings such as AWS Data Exchange \cite{AWS}, Data Republic \cite{DataRepublic}, Snowflake \cite{Snowflake}, Dawex \cite{Dawex}, Google Cloud Analytics Hub \cite{Google}, and Alibaba Cloud Data Marketplace \cite{Alibaba} provide rich compliance, access controls, and turnkey B2B data procurement, but remain single-authority systems without on-chain auditability or Byzantine fault tolerance.

\subsection{Data Marketplace Frameworks}
\cite{CastroDSC} formalizes incentive-compatible protocols for consortium data-sharing markets, but operates under permissioned governance and does not exploit public ledgers for Byzantine resilience.
\cite{Hardjono2019} outlines an “innovation-driven data marketplace” that stresses governance frameworks, yet omits any privacy-preserving compute primitives or adaptive auction mechanisms.

\subsection{Blockchain-Backed IoT and Consortium Platforms}
\cite{Dong} build an Ethereum-based IoT data-trading platform with smart-contract escrow and periodic checkpoints, yet lack privacy-preserving computation and dynamic incentive schemes.
\cite{banerjee2020} implement a Hyperledger Fabric content marketplace ensuring fair payment distribution, but remain in a permissioned setup without scalable off-chain computation or statistical consensus for adversarial settings.
\cite{Ho2021} (“Morse”) provides a blockchain-based data exchange for secure transactions but does not integrate off-chain statistical consensus to guard against malicious participants.

\subsection{Federated-Learning-Centric Marketplaces}
\cite{li2023martfl} (“martFL”) couples robust federated learning with verifiable incentives, yet assumes a central aggregator and omits any off-chain Byzantine consensus.
\cite{liu2021dealer} (“Dealer”) offers an end-to-end model marketplace with differential privacy guarantees, but relies on a central service.
\cite{Liu2022} (“PiLi”) leverages TEEs within a permissioned blockchain for secure data trading, without support for fully permissionless auctions or multi-round aggregation protocols.

\subsection{Towards Fully Decentralized FL and Resource Markets}
\cite{li2021blade} (“BLADE-FL”) introduces a blockchain-assisted framework for decentralized federated learning with optimal resource allocation, yet does not address the broader data-market architecture or on-chain revenue sharing.
\cite{s2025unifyflenablingdecentralizedcrosssilo} (“UnifyFL”) enables cross-silo federated learning with dropout resilience, but omits formal incentive mechanisms, auction protocols, and byzantine-tolerant aggregation among FL peers.

Each of these systems contributes key elements—secure auctions, escrow, TEEs, Byzantine-resilient aggregation, or resource-efficient FL—but none present an end-to-end, permissionless framework that unifies on-chain auctions, off-chain adaptive federated learning with formal aggregation scoring, and provably incentive-compatible revenue sharing as achieved in \prot.

\section{Limitations and Future work}
\label{sec:futureWork} 
While \prot \ demonstrates a strong foundation for decentralized, privacy-preserving data sharing, several opportunities remain for enhancing its robustness, expressiveness, and real-world applicability.

\begin{enumerate}
    \item {\textbf{Adaptive and Strategic Bidding Defense:} Although the current bidding mechanism is efficient and transparent, it does not account for strategic behaviors such as front-running, collusion among buyers, or fake bids to inflate auction prices. In future work, we plan to incorporate cryptographic solutions such as commit-reveal schemes, verifiable delay functions (VDFs), or other auction mechanisms to eliminate front-running and enforce fair timing guarantees.}
    \item {\textbf{Evaluation on Complex, Heterogeneous Datasets:} Our current evaluation uses simple models on well-structured image datasets. Future work should extend \prot\ to multi-modal data (text, tabular, sensor), unbalanced or domain-shifted seller distributions, and more complex architectures to enable a more robust assessment.
}
    \item {\textbf{Alternate Consensus Mechanisms and Layer-2 Optimizations:} \prot\ currently uses Ethereum for smart contracts and relies on YODA and MIRACLE for consensus. We plan to explore the use of alternate consensus protocols and Layer-2 scaling solutions to reduce on-chain latency and gas costs while preserving verifiability and finality.}
    \item {\textbf{Real-world Deployment and Integration with IoT/Edge Networks:} We further plan to deploy \prot\ on an edge/IoT testbed, with sellers as sensor owners providing private data and buyers as analytic service providers. This setup enables evaluating \prot\ under realistic bandwidth constraints, device heterogeneity, and dynamic seller churn.
}
    \item {\textbf{Model Auto-Specification Based on Buyer Context:} Currently, buyers must specify the model architecture and training setup in each bid. A natural extension is a model recommendation or synthesis engine that, given a buyer’s use case, data needs, and budget, automatically proposes a suitable training task and architecture, leveraging transfer learning and meta-learning for plug-and-play model selection.
}
\end{enumerate}

\bibliographystyle{ieeetr}
\bibliography{references}

\end{document}